\documentclass[12pt,twoside]{article}
\pdfoutput=1
\usepackage[round,longnamesfirst]{natbib}
\usepackage[noend]{algorithmic}
\usepackage{algorithm}
\usepackage[bookmarksnumbered=true,backref=page,hyperfigures=true]{hyperref}
\hypersetup{
pdfauthor = {Pavel N. Krivitsky},
pdftitle = {Exponential-Family Random Graph Models for Valued Networks}
}

\usepackage{csquotes,amsmath, amsthm, amssymb, amsbsy,array,graphicx,enumitem,subfig}

\usepackage{cancel}
\usepackage{setspace}
\usepackage{url,doi}
\usepackage{xcolor}
\usepackage{verbatim}

\newcommand{\notes}[1]{\textbf{#1}}

\def\y{\boldsymbol{y}}
\def\Y{\boldsymbol{Y}}
\def\x{\boldsymbol{x}}

\def\cnmap{\boldsymbol{\eta}}
\def\nnatpar{{p}}
\def\curvpar{\boldsymbol{\theta}}
\def\curvpars{\boldsymbol{\Theta}}
\def\natcurvpars{\boldsymbol{\Theta}_{\text{N}}}
\def\designpar{\boldsymbol{\psi}}
\def\ncurvpar{q}
\def\genstat{\boldsymbol{g}}

\def\Design{\boldsymbol{D}}
\def\design{\boldsymbol{d}}
\def\dyadvals{\mathbb{S}}

\def\changeij{\boldsymbol{\Delta}\sij}
\def\setsub{\backslash}

\DeclareMathOperator{\E}{E}

\DeclareMathOperator{\Geometric}{Geometric}

\DeclareMathOperator{\RandomChoose}{RandomChoose}

\def\dysY{\mathbb{Y}}
\def\netsY{\mathcal{Y}}
\def\iid{{\stackrel{\mathrm{i.i.d.}}{\sim}}}
\def\ind{{\stackrel{\mathrm{ind.}}{\sim}}}

\DeclareMathOperator{\Poisson}{Poisson}

\DeclareMathOperator{\Uniform}{Uniform}
\DeclareMathOperator{\Prob}{Pr}
\DeclareMathOperator{\Lik}{L}

\def\M{P}
\def\h{h}
\def\Mref{\M_\h}
\def\Mtheg{\M_{\curvpar;\Mref,\cnmap,\genstat}}

\def\sigY{\mathsf{Y}}

\def\Pteg{\Prob_{\curvpar;\cnmap,\genstat}}
\def\Eteg{\E_{\curvpar;\cnmap,\genstat}}

\def\Ptheg{\Prob_{\curvpar;\h,\cnmap,\genstat}}
\def\Etheg{\E_{\curvpar;\h,\cnmap,\genstat}}

\def\interior{^\text{o}}
\def\normc{\kappa}
\def\ceg{\normc_{\cnmap,\genstat}}

\def\cheg{\normc_{\h,\cnmap,\genstat}}

\def\cHeg{\normc_{\Mref,\cnmap,\genstat}}

\DeclareMathOperator{\ilogit}{logit^{-1}}
\def\RR{\mathbb{R}}
\def\NN{\mathbb{N}}

\def\pij{{(i,j)}}
\def\pji{{(j,i)}}
\def\ipjp{{i',j'}}
\def\pipjp{{(i',j')}}

\def\ijdysY{{\pij\in\dysY}}
\def\ipjpdysY{{\pipjp\in\dysY}}
\def\ynetsY{{\y\in\netsY}}
\def\ypnetsY{{\y'\in\netsY}}
\def\sij{_{i,j}}
\def\sijk{_{i,j,k}}
\def\sji{_{j,i}}
\def\sipjp{_{i',j'}}
\def\Yij{\Y\!\sij}
\def\yij{\y\sij}
\def\yipjp{\y\sipjp}
\def\Yji{\Y\!\sji}
\def\yji{\y\sji}

\def\Yyij{\Yij=\yij}
\def\xij{\x\sij}
\def\xijk{\x\sijk}
\def\Yy{\Y=\y}
\def\sobs{_{\text{obs}}}

\def\Yobs{\Y\sobs}
\def\yobs{\y\sobs}

\def\Yyobs{\Yobs=\yobs}

\def\half{\frac{1}{2}}
\newcommand{\natpar}[1][]{\cnmap#1(\curvpar)}

\newcommand{\natparm}[1]{\cnmap(\curvpar#1)}

\newcommand{\myexp}[1]{\exp\left(#1\right)}

\newcommand{\I}[1]{1_{#1}}
\newcommand{\pkg}[1]{\texttt{#1}}
\newcommand{\proglang}[1]{\textsf{#1}}

\providecommand{\abs}[1]{\left\lvert#1\right\rvert}

\def\t{^{\mathsf{T}}}

\newcommand{\innerprod}[2]{{#1}\cdot{#2}}

\def\Mct{\mu}

\def\dthegdref{\frac{d\Mtheg}{d\Mref}}

\def\wergm{valued ERGM}

\def\wergms{valued ERGMs}
\def\Wergms{Valued ERGMs}

\newcommand{\coefsig}[2]{$\mathbf{#1}$ ($#2$)}
\newcommand{\coefnsig}[2]{$#1$ ($#2$)}
\newcommand{\centercell}[1]{\multicolumn{1}{c}{#1}}

\theoremstyle{plain}
\newtheorem{thm}{Theorem}[section]

\title{Exponential-Family Random Graph Models for Valued Networks}

\author{Pavel N. Krivitsky}

\date{}

\begin{document}
\maketitle

\begin{abstract}
  Exponential-family random graph models (ERGMs) provide a principled
  and flexible way to model and simulate features common in social
  networks, such as propensities for homophily, mutuality, and
  friend-of-a-friend triad closure, through choice of model terms
  (sufficient statistics). However, those ERGMs modeling the more
  complex features have, to date, been limited to binary data:
  presence or absence of ties. Thus, analysis of valued networks, such
  as those where counts, measurements, or ranks are observed, has
  necessitated dichotomizing them, losing information and introducing
  biases.

  In this work, we generalize ERGMs to valued networks. Focusing on
  modeling counts, we formulate an ERGM for networks whose ties are
  counts and discuss issues that arise when moving beyond the binary
  case. We introduce model terms that generalize and model common
  social network features for such data and apply these methods to a
  network dataset whose values are counts of interactions.

  \noindent\textbf{Keywords}: {p-star model; transitivity; weighted
    network; count data; maximum likelihood estimation;
    Conway--Maxwell--Poisson distribution}
\end{abstract}
\section{Introduction}

Networks are used to represent and analyze phenomena ranging from
sexual partnerships \citep{morris1997cps}, to advice giving in an
office \citep{lazega1999mge}, to friendship relations
\citep{goodreau2008bff,newcomb1961ap}, to international relations
\citep{ward2007ppi}, to scientific collaboration, and many other
domains \citep{goldenberg2009ssn}.  More often than not, the relations
of interest are not strictly dichotomous in the sense that all present
relations are effectively equal to each other. For example, in sexual
partnership networks, some ties are short-term while others are
long-term or marital; friendships and acquaintance have degrees of
strength, as do international relations; and while a particular
individual seeking advice might seek it from some coworkers but not
others, he or she will likely do it in some specific order and weight
advice of some more than others.

Network data with valued relations come in many forms. Observing
messages \citep{freeman1980scs,diesner2005ecn}, instances of personal
interaction \citep{bernard19791980ias}, or counting co-occurrences or
common features of social actors \citep{zachary1977ifm,batagelj2006pd}
produce relations in the form of counts. Measurements, such as
duration of interaction \citep{wyatt2009dmn} or volume of trade
\citep{westveld2011mem} produce relations in the form of (effectively)
continuous values. Observations of states of alliance and war
\citep{read1954cch} produce signed relationships. Sociometric surveys
often produce ranks in addition to binary measures of affection
\citep{sampson1968npc,newcomb1961ap,bernard19791980ias,harris2003nls}.

Exponential-family random graph models (ERGMs) are generative models
for networks which postulate an exponential family over the space of
networks of interest \citep{holland1981efp,frank1986mg}, specified by
their sufficient statistics \citep{morris2008ser}, or, as with
\citet{frank1986mg}, by their conditional independence structure
leading to sufficient statistics \citep{besag1974sis}. These
sufficient statistics typically embody the features of the network of
interest that are believed to be significant to the social process
which had produced it, such as degree distribution (e.g., propensity
towards monogamy in sexual partnership networks), homophily (i.e.,
\enquote{birds of a feather flock together}), and triad-closure bias
(i.e., \enquote{a friend of a friend is a friend})
. \citep{morris2008ser}

A major limitation of ERGMs to date has been that they have been
applied almost exclusively to binary relations: a relationship between
a given actor $i$ and a given actor $j$ is either present or
absent. This is a serious limitation: valued network data have to be
dichotomized for ERGM analysis, an approach which loses information
and may introduce biases. \citep{thomas2011vtt}

Some extensions of ERGMs to specific forms of valued ties have been
formulated: to networks with polytomous tie values, represented as a
constrained three-way binary array by \citet{robins1999lml} and more
directly by \citeauthor{wyatt2009dmn}
\citeyearpar{wyatt2009dmn,wyatt2010dlr}; to
multiple binary networks by \citet{pattison1999lml}; and the authors
are also aware of some preliminary work by \citet{handcock2006sem} on
ERGMs for signed network data. \citet{rinaldo2009gde} discussed binary
ERGMs as a special case and a motivating application of their
developments in geometry of discrete exponential families.

A broad exception to this limitation has been a subfamily of ERGMs
that have the property that the ties and their values are
stochastically independent given the model parameters. Unlike the
dependent case, the likelihoods for these models have can often be
expressed as generalized linear or nonlinear models, and they tend to
have tractable normalizing constants, which allows them to more easily
be embedded in a hierarchical framework. Thus, to represent common
properties of social networks, such as actor heterogeneity,
triad-closure bias, and clustering, latent class and position models
have been used and extended to valued
networks. \citep{hoff2005bme,krivitsky2009rdd,mariadassou2010uls}

In this work, we generalize the ERGM framework to directly model
valued networks, particularly networks with count dyad values, while
retaining much of the flexibility and interpretability of binary
ERGMs, including the above-described property in the case when tie
values are independent under the model. In Section~\ref{sec:binary} we
review conventional ERGMs and describe their traits that \wergms{}
should inherit. In Section~\ref{sec:count}, we describe the framework
that extends the model class to networks with counts as dyad values
and discuss additional considerations that emerge when each dyad's
sample space is no longer binary. In Section~\ref{sec:inference} we
give some details and caveats of our implementation of these models
and briefly address the issue of ERGM degeneracy as it pertains to
count data. Applying ERGMs requiers one to specify and interpret
sufficient statistics that embody network features of interest, all
the while avoiding undesirable phenomena such as ERGM
degeneracy. Thus, in Section~\ref{sec:ctdata}, we introduce and
discuss statistics to represent a variety of features commonly found
in social networks, as well as features specific to networks of
counts. In Section~\ref{sec:example} we use these statistics to model
social forces that affect the structure of a network of counts of
social contexts of interactions among members of a divided karate club
and a network of counts of conversations among members of a
fraternity. Finally, in Section~\ref{sec:discussion}, we discuss
generalizing ERGMs to other types of valued data.

\section{\label{sec:binary}ERGMs for binary data}
In this section, we define notation, review the (potentially curved)
exponential-family random graph model and identify those of its
properties that we wish to retain when generalizing.

\subsection{Notation and binary ERGM definition}
Let $N$ be the set of actors in the network of interest, assumed known
and fixed for the purposes of this paper, and let $n\equiv\abs{N}$ be
its cardinality, or the number of actors in the network. For the
purposes of this paper, let a \emph{dyad} be defined as a (usually
distinct) pair of actors, ordered if the network of interest is
directed, unordered if not, between whom a relation of interest may
exist, and let $\dysY$ be the set of all dyads. More concretely, if
the network of interest is directed, $\dysY\subseteq N\times N$, and
if it is not, $\dysY\subseteq \{\{i,j\}: \pij \in N\times N\}$. In
many problems, a relation of interest cannot exist between an actor
and itself (e.g., a friendship network), or actors are partitioned into
classes with relations only existing between classes (e.g., bipartite
networks of actors attending events), in which case $\dysY$ is a
proper subset of $N\times N$, excluding those pairs $\pij$ between
which there can be no relation of interest.

Further, let the set of possible networks of interest (the sample
space of the model) $\netsY\subseteq 2^\dysY$, the power set of the
dyads in the network. Then a network $\ynetsY$, can be considered a
set of ties $\pij$. Again, in some problems, there may be additional
constraints on $\netsY$. A common example of such constraints are
degree constraints induced by the survey format
\citep{harris2003nls,goodreau2008bff}.

Using notation similar to that of \citet{hunter2006ice} and
\citet{krivitsky2011ans}, define an exponential family random graph
model to have the form
\begin{equation}
  \Pteg(\Yy|\x)=\frac{\myexp{\innerprod{\natpar}{\genstat(\y;\x)}}}{\ceg(\curvpar;\x)},\ \ynetsY, \label{eq:ergm}
\end{equation}
for random network variable $\Y$ and its realization $\y$; model
parameter vector $\curvpar\in\curvpars$ (for parameter space
$\curvpars\subseteq\RR^\ncurvpar$) and its mapping to canonical
parameters $\cnmap:\curvpars\to\RR^\nnatpar$; a vector of sufficient
statistics $\genstat:\netsY\to\RR^\nnatpar$, which may also depend on
data $\x$, assumed fixed and known; and a normalizing constant (in
$\y$) $\ceg:\RR^\ncurvpar\to\RR$ which ensures that \eqref{eq:ergm}
sum to 1 and thus has the value
\[\ceg(\curvpar;\x)=\sum_{\ypnetsY}\myexp{\innerprod{\natpar}{\genstat(\y';\x)}}.\]
Here, we have given the most general case defined by
\citet{hunter2006ice}: a frequently used special case is
$\ncurvpar=\nnatpar$ and $\natpar=\curvpar$, so the exponential
family is linear. For notational simplicity, we will omit $\x$ for the
remainder of this paper, as $\genstat$ incorporates it implicitly.

\subsection{Properties of binary ERGM}
\subsubsection{\label{sec:binary-changestat}Conditional distributions and change statistics}
\citet{snijders2006nse}, \citet{hunter2008epf},
\citet{krivitsky2011ans}, and others define \emph{change statistics},
which emerge when considering the probability of a single dyad having
a tie given the rest of the network and provide a convenient local
interpretation of ERGMs. To summarize, define the $\nnatpar$-vector of change
statistics
\[\changeij \genstat(\y)\equiv \genstat(\y+\pij)-\genstat(\y-\pij),\]
where $\y+\pij$ is the network $\y$ with edge or arc $\pij$ added if
absent (and unchanged if present) and $\y-\pij$ is the network $\y$
with edge or arc $\pij$ removed if present (and unchanged if absent). Then, through cancellations,
\begin{equation*}
\Pteg(\Yij=1|\Y-\pij =\y-\pij)=\ilogit\left(\innerprod{\natpar}{\changeij \genstat(\y)}\right).
\end{equation*}
It is often the case that the form of $\changeij \genstat(\y)$ is
simpler than that of $\genstat(\y)$ both algebraically and
computationally. For example, the change statistic for edge count
$\abs{\y}$ is simply $1$, indicating that a unit increase in
$\natpar[_{\abs{\y}}]$ will increase the conditional log-odds of a tie
by $1$, while the change statistic for the number of triangles in a
network is $\abs{\y_i\cap\y_j}$, the number of neighbors $i$ and $j$
have in common, suggesting that, a positive coefficient on this
statistic will increase the odds of a tie between $i$ and $j$
exponentially in the number of common neighbors.
\citet{hunter2008epf} and \citet{krivitsky2011ans} offer a further
discussion of change statistics and their uses, and
\citet{snijders2006nse} and \citet{schweinberger2011isd} use them to
diagnose degeneracy in ERGMs. It would be desirable for a
generalization of ERGM to valued networks to facilitate similar local
interpretation.

Furthermore, the conditional distribution serves as the basis for
maximum pseudo-likelihood estimation (MPLE) for these
models. \citep{strauss1990pes}

\subsubsection{\label{sec:logit-reg}Relationship to logistic regression}
If the model has the property of \emph{dyadic independence} discussed
in the Introduction, or, equivalently, the change statistic
$\changeij \genstat(\y)$ is constant in $\y$ (but may vary for
different $\pij$), the model trivially reduces to logistic
regression. In that case, the MLE and the MPLE are
equivalent. \citep{strauss1990pes} Similarly, it may be a desirable
trait for valued generalizations of ERGMs to also reduce to GLM for
dyad-independent choices of sufficient statistics.

\section{\label{sec:count}ERGM for counts}
We now define ERGMs for count data and discuss the issues that arise
in the transition.
\subsection{Model definition}
Define $N$, $n$, and $\dysY$ as above. Let $\NN_0$ be the set of
natural numbers and 0. Here, we focus on counts with no \emph{a
  priori} upper bound --- or counts best modeled thus. Instead of
defining the sample space $\netsY$ as a subset of a power set, define
it as $\netsY\subseteq \NN_0^\dysY$, a set of mappings that assign to
each dyad $\ijdysY$ a count. Let $\yij=\y\pij \in \NN_0$ be the value
associated with dyad $\pij$.

A (potentially curved) ERGM for a random network of counts
$\Y\in\netsY$ then has the pmf
\begin{equation}
  \Ptheg(\Yy)=\frac{h(\y)\myexp{\innerprod{\natpar}{\genstat(\y)}}}{\cheg(\curvpar)},\label{eq:ctergm}
\end{equation}
 where the normalizing constant
\[\cheg(\curvpar)=\sum_\ynetsY h(\y)\myexp{\innerprod{\natpar}{\genstat(\y)}},\]
with $\cnmap$, $\genstat$, and $\curvpar$ defined as above, and
\begin{equation}
  \curvpars\subseteq\natcurvpars=\{\curvpar'\in\RR^\ncurvpar:\cheg(\curvpar')<\infty\}\label{eq:ctergm-Theta}
\end{equation}
\citetext{\citealp[pp.~115--116]{barndorffnielsen1978ief};
  \citealp[pp.~1--2]{brown1986fse}}, with $\natcurvpars$ being the
natural parameter space.  Notably, constraint \eqref{eq:ctergm-Theta}
is trivial for binary networks, since their sample space is finite,
whereas for counts \eqref{eq:ctergm-Theta} may constitute a relatively
complex constraint.

For the remainder of this paper, we will focus on linear ERGMs, so
unless otherwise noted, $\nnatpar=\ncurvpar$ and
$\natpar\equiv\curvpar$.

\subsection{Reference measure}
In addition to the specification of the sufficient statistics
$\genstat$ and, for curved families, mapping $\cnmap$ of model
parameters to canonical parameters, an ERGM for counts depends on the
specification of the function $\h:\netsY\to [0,\infty)$. Formally,
along with the sample space, it specifies the reference measure: the
distribution relative to which the exponential form is specified. For
binary ERGMs, $\h$ is usually not specified explicitly, though in some
ERGM applications, such as models with offsets \citep[for
example]{krivitsky2011ans} and profile likelihood calculations of
\citet{hunter2008epf}, the terms with fixed parameters are implicitly
absorbed into $\h$.

For valued network data in general, and for count data in particular,
specification of $\h$ gains a great deal of importance, setting the
baseline shape of the dyad distribution and constraining the parameter
space. Consider a very simple $\nnatpar=1$ model with
$\genstat(\y)=(\sum_{\ijdysY}\yij)$, the sum of all dyad values. If
$\h(\y)=1$ (i.e., discrete uniform), the resulting family has the pmf
\[
  \Ptheg(\Yy)=\frac{\myexp{\curvpar\sum_{\ijdysY}\yij}}{\cheg(\curvpar)}=\prod_{\ijdysY}\frac{\myexp{\curvpar\yij}}{1-\myexp{\curvpar}},
\]
giving the dyadwise distribution $\Yij\iid
\Geometric(p=1-\myexp{\curvpar})$, with $\curvpar<0$ by
\eqref{eq:ctergm-Theta}.  On the other hand, suppose that, instead,
$\h(\y)=\prod_{\ijdysY}(\yij!)^{-1}$. Then,
\[
  \Ptheg(\Yy)=\frac{\myexp{\curvpar\sum_{\ijdysY}\yij}}{\cheg(\curvpar)\prod_{\ijdysY}\yij!}=\prod_{\ijdysY}\frac{\myexp{\curvpar\yij}}{\yij!\myexp{\curvpar}},
\]
giving $\Yij\iid\Poisson(\mu=\myexp{\curvpar})$, with
$\natcurvpars=\RR$. The shape of the resulting distributions for a
fixed mean is given in Figure~\ref{fig:h-effect}.

\begin{figure}
  \begin{center}
    \includegraphics[width=0.7\textwidth]{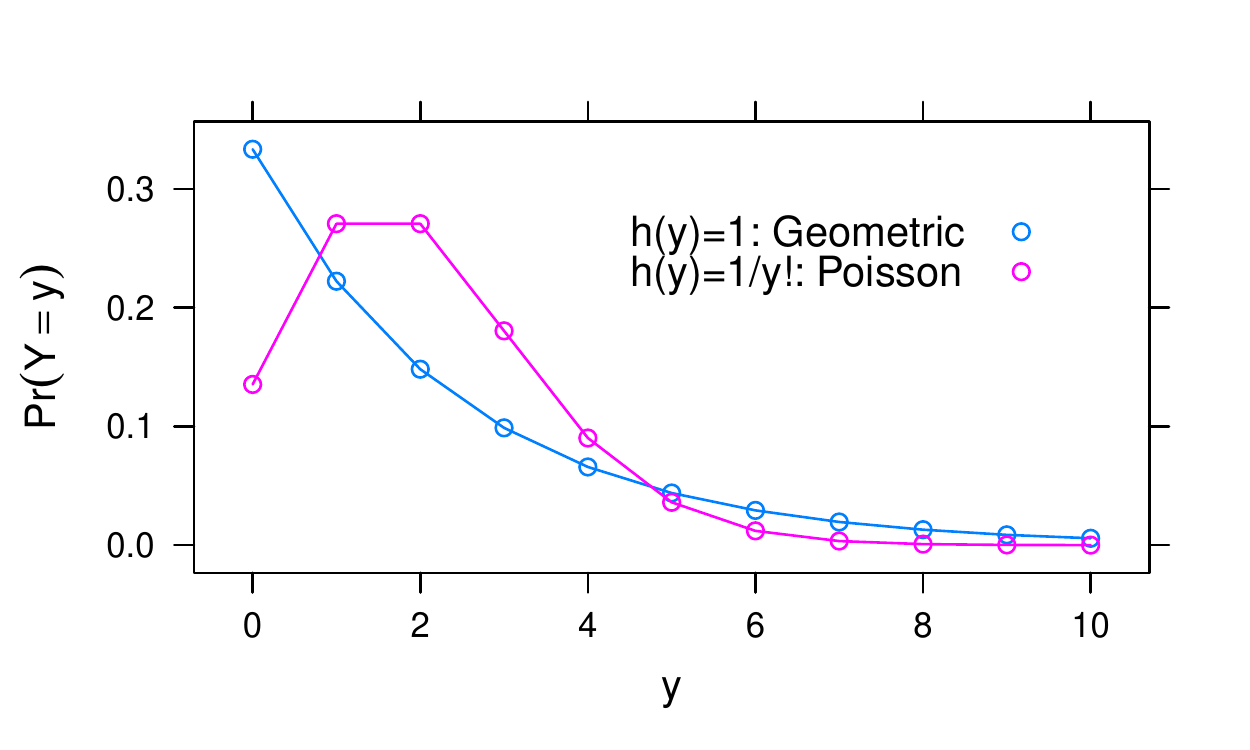}
  \end{center}
  \caption{\label{fig:h-effect}Effect of $\h$ on the shape of
    the distribution. (The mean is fixed at 2.) }
\end{figure}
The reference measure $\h$ thus determines the support and the basic
shape of the ERGM distribution. For this reason, we define a
\emph{geometric-reference ERGM} to have the form \eqref{eq:ctergm}
with $\h(\y)=1$ and a \emph{Poisson-reference ERGM} to have
$\h(\y)=\prod_{\ijdysY}(\yij!)^{-1}$.

Note that this does not mean that any Poisson-reference ERGM will, even under
dyadic independence, be dyadwise Poisson. We discuss the sufficient
conditions for this in Section~\ref{sec:pois-reg}.


\section{\label{sec:inference}Inference and implementation}
As exponential families, \wergms{}, and ERGMs for counts in
particular, inherit the inferential properties of discrete exponential
families in general and binary ERGMs in particular, including
calculation of standard errors and analysis of deviance. They also
inherit the caveats. For example, the Wald test results based on
standard errors depend on asymptotics which are questionable for ERGMs
with complex dependence structure \citep{hunter2006ice}, so we confirm
the most important of the results using a simple Monte Carlo test: we
fit a nested model without the statistic of interest and simulate its
distribution under such a model. The quantile of the observed value of
the statistic of interest can then be used as a more robust $P$-value.

At the same time, generalizing ERGMs to counts raises additional
inferential issues. In particular, the infinite sample space of counts
means that the constraint \eqref{eq:ctergm-Theta} is not always
trivially satisfied, which results in some \wergm{} specifications not
fulfilling regularity conditions. We give an example of this in
Section~\ref{sec:dispersion} and Appendix~\ref{app:CMP-steep}. These
issues also lead to additional computational issues.

\subsection{Computational issues}
The greatest practical difficulty associated with likelihood inference
on these models is usually that the normalizing constant
$\cheg(\curvpar)$ is intractable, its exact evaluation requiring
integration over the sample space $\netsY$. However, the
exponential-family nature of model also means that, provided a method
exists to simulate realizations of networks from the model of interest
given a particular $\curvpar$, the methods of \citet{geyer1992cmc} for
fitting exponential families with intractable normalizing constants
and, more specifically, their application to ERGMs by
\citet{hunter2006ice}, may be used. These methods rely on network
sufficient statistics rather than networks themselves and can thus be
used with little modification. More concretely, the ratio of two
normalizing constants evaluated at $\curvpar'$ and $\curvpar$ can be
expressed as
\begin{align*}
  \frac{\cheg(\curvpar')}{\cheg(\curvpar)}&=\frac{\sum_\ynetsY \h(\y)\myexp{\innerprod{\natparm{'}}{\genstat(\y)}}}{\cheg(\curvpar)}\\
  &=\frac{\sum_\ynetsY \h(\y)\myexp{\innerprod{\left(\natparm{'}-\natpar\right)}{\genstat(\y)}}\myexp{\innerprod{\natpar}{\genstat(\y)}}}{\cheg(\curvpar)}\\
  &=\sum_\ynetsY \myexp{\innerprod{\left(\natparm{'}-\natpar\right)}{\genstat(\y)}}\frac{\h(\y)\myexp{\innerprod{\natpar}{\genstat(\y)}}}{\cheg(\curvpar)}\\
  &=\Etheg\left(\myexp{\innerprod{\left(\natparm{'}-\natpar\right)}{\genstat(\Y)}}\right),
\end{align*}
so given a sample $\Y^{(1)},\dotsc,\Y^{(S)}$ from an initial guess
$\curvpar$, it can be estimated
\[\frac{\cheg(\curvpar')}{\cheg(\curvpar)}\approx\sum_{s=1}^S\myexp{\innerprod{\left(\natparm{'}-\natpar\right)}{\genstat(\Y^{(s)})}}.\]

Another method for fitting ERGMs, taking advantage of the equivalence
of the method of moments to the maximum likelihood estimator for
linear exponential families, was implemented by
\citet{snijders2002mcm}, using the algorithm by \citet{robbins1951sam}
for simulated statistics to fit the model. This approach also
trivially extends to \wergms.

Furthermore, because the normalizing constant (if it is finite) is
thus accommodated by the fitting algorithm, we may focus on the
unnormalized density for the purposes of model specification and
interpretation. Therefore, for the remainder of this paper, we specify
our models up to proportionality, as \citet{geyer1999lis} suggests.

That \eqref{eq:ctergm-Theta} is not trivially satisfied for all
$\curvpar\in \RR^q$ presents an additional computational challenge:
even for relatively simple network models, $\curvpars$ may have a
nontrivial shape. For example, even a simple geometric-reference ERGM
\[
\Ptheg(\Yy)\propto\prod_{\ijdysY}\myexp{\innerprod{\curvpar}{\xij \yij}},
\]
a geometric GLM with a covariate $\nnatpar$-vector $\xij$,
has parameter space
\[\curvpars=\{\curvpar'\in \RR^\nnatpar:\forall_{\ijdysY}\curvpar\cdot\xij<0\},\]
an intersection of up to $\abs{\dysY}$ half-spaces (linear
constraints). Models with complex dependence structure may have less
predictable parameter spaces, and, due to the nature of the algorithm
of \citet{hunter2006ice}, the only general way to detect whether a
guess for $\curvpar$ had strayed outside of $\curvpars$ may be by
diagnostics on the simulation. Bayesian inference with improper priors
faces a similar problem, and addressing it in the context of ERGMs is
a subject for future work. For this paper, we focus on models in which
parameter spaces are provably unconstrained or have very simple
constraints.

We base our implementation on the \proglang{R} package \pkg{ergm} for
fitting binary ERGMs. \citep{handcock2010epf} The design of that
package separates the specification of model sufficient statistics
from the specification of the sample space of networks
\citep{hunter2008epf}, and so we implement our models by substituting
in a Metropolis-Hastings sampler that implements our $\netsY$ and $\h$
of interest. (A simple sampling algorithm for realizations from a
Poisson-reference ERGM, optimized for zero-inflated data, is described
in Appendix~\ref{app:pois-MH}.)  This implementation will be
incorporated into a future public release of \pkg{ergm}.

\subsection{\label{sec:degeneracy}Model degeneracy}
Application of ERGMs has long been associated with a complex of
problems collectively referred to as
\enquote{degeneracy}. \citep{handcock2003ads,rinaldo2009gde,schweinberger2011isd}
\citeauthor{rinaldo2009gde}, in particular, list three specific,
interrelated, phenomena: 1) when a parameter configuration --- even
the MLE --- induces a distribution for which only a small number of
possible networks have non-negligible probabilities, and these
networks are often very different from each other (e.g., a
sparser-than-observed graph and a complete graph) for an effectively
bimodal distribution; 2) when the MLE is hard to find by the available
MCMC methods; and 3) when the probability of the observed network
under the MLE is relatively low --- the observed network is,
effectively, between the modes. This bimodality and concentration is
often a consequence of the model inducing overly strong positive
dependence among dyad values. For example, \citet{snijders2006nse} use
change statistics to show that under models with positive coefficients
on triangle and $k$-star ($k\ge 2$) counts --- the classic
\enquote{degenerate} ERGM terms --- every tie added to the network
increases the conditional odds of several other ties and does not
decrease the odds of any, creating what \citeauthor{snijders2006nse}
call an \enquote{avalanche} toward the complete graph, which emerges
as by far the highest-probability realization. (More concretely, under
a model with a triangle count with coefficient $\curvpar_\triangle$,
adding a tie $\pij$ increases the conditional odds of as many ties as
$i$ and $j$ have neighbors by $\myexp{\curvpar_\triangle}$.) Adjusting
other parameters, such as density, down to obtain the expected level
of sparsity close to that of the observed graph merely induces the
bimodal distribution of Phenomenon~1.

An infinite sample space makes Phenomenon~1, as such, unlikely,
because the \enquote{avalanche} does not have a maximal graph in which
to concentrate. However, it does not preclude excessive dependence
inducing a bimodal distribution at the MLE, even if neither mode is
remotely degenerate in the probabilistic sense. The observed network
being between these modes, this may lead to Phenomenon~3, and, due
to the nature of the estimation algorithms, such a situation may,
indeed, lead to failing estimation --- Phenomenon~2.

In this work, we seek to avoid this problem by constructing statistics
that prevent the \enquote{avalanche} by limiting dependence or
employing counterweights to reduce it. (An example of the former
approach is the modeling of transitivity in
Section~\ref{sec:triad-closure}, and an example of the latter is the
centering in the within-actor covariance statistic developed in
Section~\ref{sec:heterogeneity}.)  Formal diagnostics developed to
date, such as those of \citet{schweinberger2011isd} do not appear to
be directly applicable to models with infinite sample spaces, so we rely
on MCMC diagnostics \citep{goodreau2008st} instead.

\section{\label{sec:ctdata}Statistics and interpretation for count data}
In this section, we develop sufficient statistics for count data to
represent network features that may be of interest and discuss their
interpretation. In particular, unless otherwise noted, we focus on the
Poisson-reference ERGM without complex constraints:
$\netsY=\NN_0^\dysY$ and $\h(\y)=\prod_{\ijdysY}(\yij!)^{-1}$.

\subsection{Interpretation of model parameters}
The sufficient statistics of the binary ERGMs and \wergms{} alike
embody the structural properties of the network that are of
interest. The tools available for interpreting them are similar as
well.

\subsubsection{\label{sec:expectation}Expectations of sufficient statistics}
In a linear ERGM, if $\natcurvpars$ is an open set, then, for every
$k\in 1..\nnatpar$, and holding $\curvpar_{k'}$, $k'\ne k$, fixed, it
is a general exponential family property that the expectation
$\Etheg(\genstat_k(\Y))$ is strictly increasing in
$\curvpar_k$. \citep[pp. 120--121]{barndorffnielsen1978ief} Thus, if
the statistic $\genstat_k$ is a measurement of some feature of
interest of the network (e.g., magnitude of counts, interactions
between or within a group, isolates, triadic structures), a greater
value of $\curvpar_k$ results in a distribution of networks with more
of the feature measured by $\genstat_k$ present.

\subsubsection{\label{sec:disc-changestat}Discrete change statistic and conditional distribution}
Binary ERGM statistics have a \enquote{local} interpretation in the
form of \emph{change statistics} summarized in
Section~\ref{sec:binary-changestat}, we describe similar tools for
\enquote{local} interpretation of ERGMs for counts here.

Define the set of networks
\[\netsY\sij(\y)\equiv\{\y'\in\netsY:\forall_{\ipjp\in \dysY\setsub\{\pij\}}\y'_\ipjp=\y_\ipjp\}.\]
That is, $\netsY\sij(\y)$ is the set of networks such that all dyads
but the focus dyad $\pij$ are fixed to their values in $\y$ while
$\pij$ itself may vary over its possible values; and define
$\y_{\pij=k}\in \{\y'\in\netsY\sij(\y): \y'\sij=k\}$ (a singleton set)
to be the network with non-focus dyads fixed and focus dyad set to
$k$. Then, let the \emph{discrete change statistic}
\[\changeij^{k_1\to k_2}\genstat(\y)\equiv\genstat(\y_{\pij=k_2})-\genstat(\y_{\pij=k_1}).\]

This statistic emerges when taking the ratio of probabilities of two
networks that are identical except for a single dyad value:
\begin{align*}
\frac{\Ptheg(\Yij=\y_{\pij=k_2}|\Y\in\netsY\sij(\y))}{\Ptheg(\Yij=\y_{\pij=k_1}|\Y\in\netsY\sij(\y))}
&=\frac{h\sij(k_2)}{h\sij(k_1)}\myexp{\innerprod{\curvpar}{\changeij^{k_1 \to k_2}(\y)}},
\end{align*}
where $\h\sij:\NN_0\to \RR$ is the component of $\h$ associated with
dyad $\pij$, such that $\h(\y)\equiv \prod_\ijdysY \h\sij(\yij)$, if
it can be thus factored. For a Poisson-reference ERGM,
$h\sij(k)=(k!)^{-1}$. This may be used to assess the effect of a
particular ERGM term on the \emph{decay rate} of the ratios of
probabilities of successive values of dyads \citep{shmueli2005udf} and
on the shape of the dyadwise conditional distribution: the conditional
distribution of a dyad $\ijdysY$, given all other dyads
$\ipjpdysY\setsub\{\pij\}$,
\begin{align*}
\Ptheg(\Yyij|\Y\in\netsY\sij(\y))
&=\frac{h\sij(\yij)\myexp{\innerprod{\curvpar}{\genstat(\y)}}}{\sum_{\y'\in \netsY\sij(\y)} \h(\y'\sij)\sij\myexp{\innerprod{\curvpar}{\genstat(\y')}}}\\
&=\frac{h\sij(\yij)\myexp{\innerprod{\curvpar}{\changeij^{k_0\to \yij}\genstat(\y)}}}{\sum_{k\in\NN_0} \h\sij(k)\myexp{\innerprod{\curvpar}{\changeij^{k_0\to k}\genstat(\y)}}},
\end{align*}
for an arbitrary baseline $k_0$.

\subsection{\label{sec:ctstats}Model specification statistics}
We now propose some specific model statistics to represent common network
structural properties and distributions of counts.

\subsubsection{\label{sec:pois-reg}Poisson modeling}
We begin with statistics that produce Poisson-distributed dyads and
model network phenomena that can be represented in a dyad-independent
manner. As a binary ERGM reduces to a logistic regression model under
dyadic independence, a Poisson-reference ERGM may reduce to a Poisson
regression model.

In a Poisson-reference ERGM, the normalizing constant has a simple
closed form if $\genstat(\y')$ is linear in $\y'\sij$ and does not
depend on any other dyads $\y'\sipjp$, $\pipjp\ne\pij$:
\begin{equation}
  \forall_\ynetsY\forall_{\y'\sij\in \NN_0}\changeij^{0\to \y'\sij}\genstat(\y)=\y'\sij\xij \label{eq:cond-pois-req-lin}.
\end{equation}
for $\xij\equiv \changeij^{k\to k+1}\genstat(\y)$ for any $k\in
\NN_0$. Then,
\[\Yij \ind \Poisson\left(\mu=\myexp{\innerprod{\curvpar}{\changeij^{0\to 1}\genstat(\y)}}\right),\]
giving a Poisson log-linear model, and $\changeij^{0\to 1}\genstat$
effectively becomes the covariate vector for $\Yij$. (If
$\genstat(\y')$ is linear in $\y'\sij$ but does depend on other dyads
--- $\xij$ in \eqref{eq:cond-pois-req-lin} depends on $\y'\sipjp$ but
not on $\y'\sij$ itself --- the dyad distribution is conditionally
Poisson but not marginally so. An example of this arises in Section~\ref{sec:mutuality}.)

\citet{morris2008ser} describe many dyad-independent sufficient
statistics for binary ERGMs. All of them have the general form
\begin{equation}\genstat_k(\y) \equiv \sum_\ijdysY \yij\xijk,\label{eq:poisreg-stat}\end{equation}
where $\xijk\equiv\changeij\genstat_k$ and $\xijk$ may be viewed as
exogenous (to the model) covariates in a logistic regression for each
tie. They could then be used to model a variety of patterns for degree
heterogeneity and mixing among actors over (assumed) exogenous
attributes. For example, for a uniform homophily model, $\xijk$ may be
an indicator of whether $i$ and $j$ belong to the same group. If
$\yij$ are counts, these statistics induce a Poisson regression type
model (for a Poisson-reference ERGM), where the effect of a unit
increase in some $\curvpar_k$ on dyad $\pij$ is to increase its
expectation by a factor of $\myexp{\xijk}$. \citet{krivitsky2009rdd}
use this type of model Slovenian periodical \enquote{co-readerships}
\citep{batagelj2006pd} --- numbers of readers who report reading each
pair of periodicals of interest --- using as exogenous covariates the
class of periodical (daily, weekly, regional, etc.) and the overall
readership levels of each periodical.

Curved (i.e., $\natpar\ne\curvpar$, $\nnatpar>\ncurvpar$, and
$\cnmap$ not a linear mapping) ERGMs, in which the $\genstat$ satisfy
\eqref{eq:cond-pois-req-lin} and dyadic independence, may induce
nonlinear Poisson regression. An example of this is the likelihood
component of some latent space network models, with latent space
positions being treated as free parameters: for example, the
likelihoods of the Poisson models of \citet{hoff2005bme} and
\citet{krivitsky2009rdd} are special cases of such an ERGM, with
$\natpar=\left(\cnmap\sij(\curvpar)\right)_\ijdysY$ and
$\genstat(\y)=\left(\yij\right)_\ijdysY$ (i.e., the sufficient
statistic is the network), and $\cnmap\sij(\curvpar)$ mapping latent
space positions and other parameters contained in $\curvpar$ to the
logarithms of dyad means (i.e., the dyadwise canonical parameters).

\subsubsection{\label{sec:count-thresholded}Zero modification}
We now turn to model terms that may reshape the distribution of the
counts away from Poisson. Social networks tend to be sparse, and
larger networks of similar nature tend to be more sparse
\citep{krivitsky2011ans}. If the interactions among the actors are
counted, it is often the case that if two actors interact at all, they
interact multiple times. This leads to dyadwise distributions that are
zero-inflated relative to Poisson.

These features of sparsity can be modeled using statistics developed
for binary ERGMs, applied to a network produced by thresholding the
counts (at $1$, for zero-modification). For example, a
Poisson-reference ERGM with $\nnatpar=2$ and
\[\genstat(\y)=\left(\sum_{\ijdysY}\yij,\sum_{\ijdysY}\I{\yij>0}\right)\t\]
has dyadwise distribution
\begin{align*}
  \Ptheg(\Yy)
  &\propto \prod_{\ijdysY}\myexp{\curvpar_1\yij+\curvpar_2\I{\yij>0}}/\yij!.
\end{align*}
This is a parametrization of a zero-modified Poisson distribution
\citep{lambert1992zip}, though not a commonly used one, with the
probability of $0$ being
$(1+\myexp{\curvpar_2}(\myexp{\myexp{\curvpar_1}}-1))^{-1}$ and
nonzero values being distributed (conditionally on not being $0$)
$\Poisson(\mu=\myexp{\curvpar_1})$, both reducing to $\Poisson$'s when
$\curvpar_2=0$. Notably, the probability of $0$ decreases as
$\curvpar_1$ increases, rather than being solely controlled by
$\curvpar_2$.

\subsubsection{\label{sec:dispersion}Dispersion modeling}
Consider the social network of face-to-face conversations among people
living in a region. A typical individual will likely not interact at
all with vast majority of others, have one-time or infrequent
interaction with a large number of others (e.g., with clerks or
tellers), and a lot of interaction with a relatively small number of
others (e.g., family, coworkers). Some of this may be accounted for by
information about social roles and preexisting relationships, but if
such information is not available, this leads to a highly
overdispersed distribution relative to Poisson, or even zero-inflated
Poisson. Overdispersed counts are often modeled using the
negative binomial distribution. \citep[p. 199]{mccullagh1989glm}
However, the negative binomial distribution with an unknown dispersion
parameter is not an exponential family, making it difficult to fit
using our inference techniques. We thus discuss two purely
exponential-family approaches for dealing with non-Poisson-dispersed
interaction counts in general and overdispersed counts in particular.

\paragraph{Conway--Maxwell--Poisson Distribution}
Conway--Maxwell--Poisson (CMP) distribution \citep{shmueli2005udf} is an
exponential family for counts, able to represent both under- and
overdispersion: adding a sufficient statistic of the form
\begin{equation}
\genstat_{\text{CMP}}(\y)=\sum_{\ijdysY}\log(\yij!), \label{eq:CMP-term}
\end{equation}
to a Poisson-reference ERGM otherwise fulfilling conditions for
Poisson regression described in Section~\ref{sec:pois-reg} turns a
Poisson regression model for dyads into a CMP regression model.

Its coefficient, $\curvpar_{\text{CMP}}$, constrained by
\eqref{eq:ctergm-Theta} to $\curvpar_{\text{CMP}}\le 1$, controls the
degree of dispersion: $\curvpar_{\text{CMP}}=0$ retains the Poisson
distribution; $\curvpar_{\text{CMP}}<0$ induces underdispersion
relative to Poisson, approaching the Bernoulli distribution as
$\curvpar_{\text{CMP}}\to -\infty$; and $\curvpar_{\text{CMP}}>0$
induces overdispersion, attaining the geometric distribution at
$\curvpar_{\text{CMP}}= 1$, its most overdispersed point.

Normally, the greatest hurdle associated with using CMP is that its
normalizing constant does not, in general, have a known closed
form. In our case, because intractable normalizing constants are
already accommodated by the methods of Section~\ref{sec:inference}, so
using CMP in this setting requires no additional effort.

At the same time, CMP is neither regular nor steep (per
Appendix~\ref{app:CMP-steep}), so the properties of its estimators are
not guaranteed, particularly for highly overdispersed data. We have
found experimentally that counts as dispersed as geometric
distribution or more so often cause the fitting methods of
Section~\ref{sec:inference} to fail.

\paragraph{Variance-like parameters}
Some control over the variance can be attained by adding a statistic
of the form $\genstat(\y)=\sum_{\ijdysY}\yij^a$, $a\ne 1$. Statistics
with $a>1$, such as $\genstat(\y)=\sum_{\ijdysY}\yij^2$, suffer the
same problem as a Strauss point process \citep{kelly1976nsm}: for any
$\theta,\epsilon>0$, $\lim_{y\to\infty} \exp(\theta
y^{1+\epsilon})/y!=\infty$, leading to \eqref{eq:ctergm-Theta}
constraining $\theta\le 0$, able to represent only
\emph{underdispersion}.

Thus, we propose to model dispersion by adding a statistic of the form
\begin{equation}
\genstat(\y)=\sum_{\ijdysY}\yij^{1/2}=\sum_{\ijdysY}\sqrt{\yij}. \label{eq:sqrtsum}
\end{equation}
To the extent that the counts are Poisson-like, the square root is a
variance-stabilizing transformation
\citep[p. 196]{mccullagh1989glm}. Then, a model with $\nnatpar=2$ and
dyadwise sufficient statistic
\begin{equation}
\genstat(\y)=\left(\sum_{\ijdysY}\sqrt{\yij}, \sum_{\ijdysY}\yij\right)\t \label{eq:sqrtsum-model}
\end{equation}
may be viewed as a modeling the first and second moments of
$\sqrt{\yij}$. That the highest-order term is still on the order of
$\yij$ guarantees that $\curvpars=\RR^\nnatpar$ --- a practical
advantage over CMP.

As with CMP, the normalizing constant is intractable. To explore the
shape of this distribution, we fixed $\curvpar_1$ at each of a range
of values and found $\curvpar_2$s such that the induced distribution
had the expected value of $1$. We then simulated from the fit. The
estimated pmf for each configuration and the comparison with the
geometric distribution with the same expectation is given in
Figure~\ref{fig:sqrt-model}. Smaller coefficients on
\eqref{eq:sqrtsum} ($\curvpar_1$) correspond to greater dispersion,
with coefficients on dyad sum ($\curvpar_2$) increasing to compensate,
and vice versa, with $\curvpar_1=0$ corresponding to a Poisson
distribution. As the dispersion increases, the mean is preserved in
part by increasing $\Prob(\Yij=0)$ and, for sufficiently high values
of $\yij$, the geometric distribution still dominates. Thus, there is
a trade-off between the convenience of a model without complex
constraints on the parameter space and the ability to model greater
dispersion. In practice, if the substantive reasons for overdispersion
are due to unaccounted-for heterogeneity, the latter might not be a
serious disadvantage, and excess zeros can be compensated for by a
term from Section~\ref{sec:count-thresholded}.
\begin{figure}
  \includegraphics[width=1\textwidth]{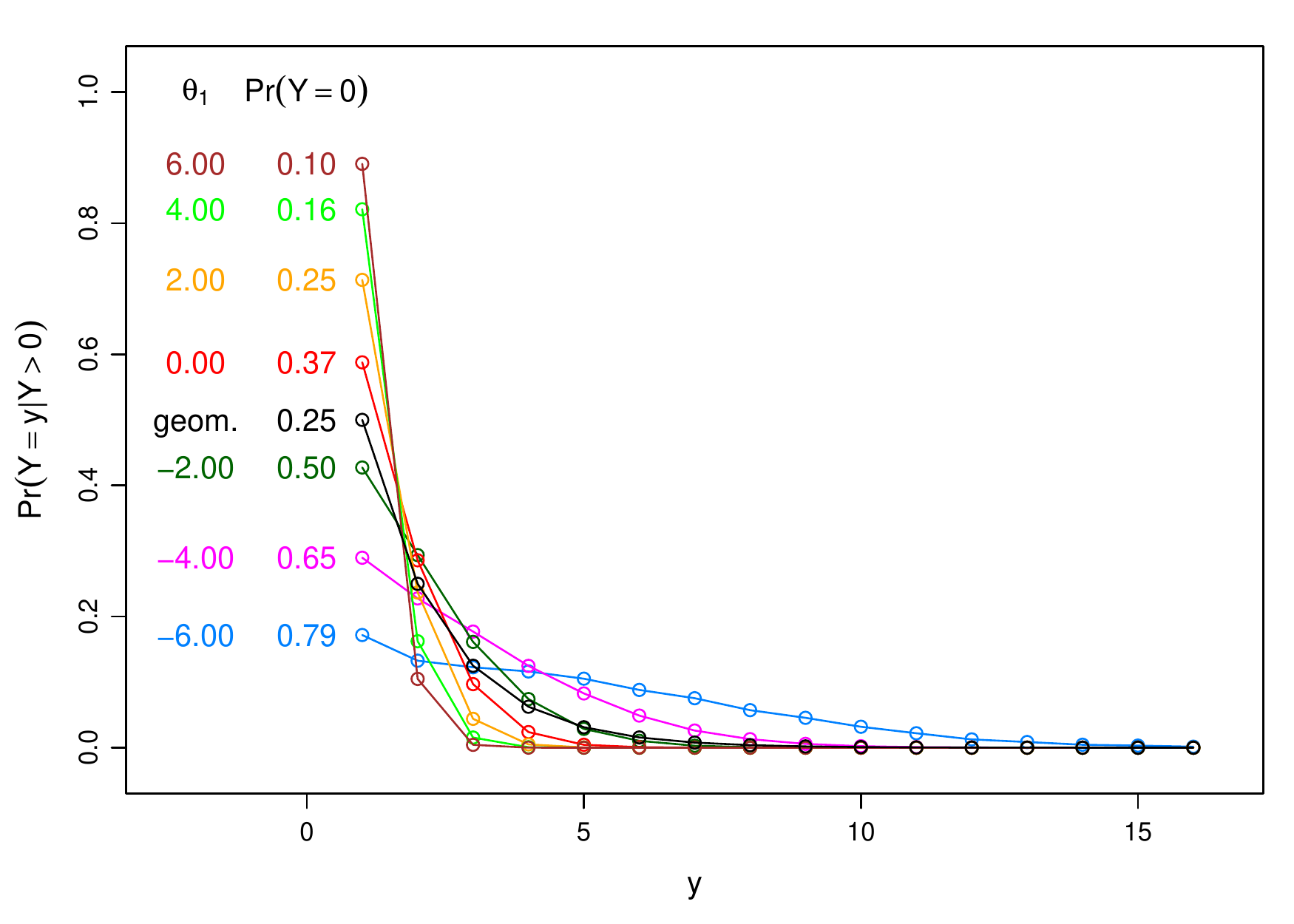}
  \caption{\label{fig:sqrt-model}Dyadwise distributions attainable by
    the model \eqref{eq:sqrtsum-model}. Because $\Prob(Y=0)$ varies
    greatly for different $\curvpar_1$ yet can be adjusted separately
    by an appropriate model term, we plot the probabilities
    conditional on $Y>0$.}
\end{figure}

\subsubsection{\label{sec:mutuality}Mutuality}
Many directed networks, such as friendship nominations, exhibit
\emph{mutuality} --- that, other things being equal, if a tie $\pij$
exists, a tie $\pji$ is more likely to exist as well --- and binary
ERGMs can model this phenomenon using a sufficient statistic
$\genstat(\y)=\sum_{\ijdysY,i<j} \yij\yji =
\sum_{\ijdysY,i<j}\min(\yij,\yji)$, counting the number of
reciprocated ties. \citep{holland1981efp} Other sufficient statistics
that can model it include $\genstat(\y)=\sum_{\ijdysY,i<j}
\I{\yij\ne\yji}$ and $\genstat(\y)=\sum_{\ijdysY,i<j} \I{\yij=\yji}$,
the counts of asymmetric and symmetric dyads,
respectively. \citep{morris2008ser}

In the presence of an edge count term, these three are simply
different parametrizations of the same distribution family:
\[\yij\yji = \frac{(\yij+\yji)-\I{\yij\ne\yji}}{2} = \frac{(\yij+\yji)-1+\I{\yij=\yji}}{2}.\]
Nevertheless, these three different statistics suggest two major ways
to generalize the terms to count data: by evaluating a product or a
minimum of the values, or by evaluating their similarity or
difference. We discuss them in turn.

\paragraph{Product}
It is tempting to model mutuality for count data in the same manner as
for binary data, with $\yij$ and $\yji$ being values rather than
indicators. For example, a simple model with overall dyad mean and
reciprocity terms, with $\nnatpar=2$ and
\[\genstat(\y)=\left(\sum_{\ijdysY}\yij,\sum_{\ijdysY,i<j}\yij\yji\right)\t\]
would have a conditional Poisson distribution:
\begin{align*}
\Yyij|\Y\in\netsY\sij(\y)
&\sim \Poisson\left(\mu=\myexp{\curvpar_1 + \curvpar_2 \yji}\right),
\end{align*}
a desirable property. However, because for any $c>0$,
$\lim_{y\to\infty} \exp(c y^2)/(y!)^2=\infty$, for $\curvpar_2>0$,
representing positive mutuality, \eqref{eq:ctergm-Theta} is not
fulfilled. (Note that the expected value of $\Yij$ is
\emph{exponential} in the value of $\Yji$ and vice versa. Again, a
Strauss point process exhibits a similar
problem. \citep{kelly1976nsm})

\paragraph{Geometric mean}
As with dispersion, the problem can be alleviated by using the
geometric mean of $\yij$ and $\yji$ instead of their product. As in
Section~\ref{sec:dispersion}, this choice may be justified as an
analog of covariance on variance-stabilized
counts. 
This changes the shape of the distribution in ways that are difficult
to interpret: if
\[\genstat(\y)=\left(\sum_{\ijdysY}\yij,\sum_{\ijdysY,i<j}\sqrt{\yij\yji}\right)\t,\] then
\[\Ptheg(\Yyij|\Y\in\netsY\sij(\y))\propto\myexp{\curvpar_1\yij+(\curvpar_2\sqrt{\yji})\sqrt{\yij}}/\yij!,\]
and, with nonzero $\yji$, the probabilities of greater values of
$\Yij$ are inflated by more. The analogy to covariance further
suggests centering the statistic:
\[\genstat(\y)=\sum_{\ijdysY,i<j}(\sqrt{\yij}-\overline{\sqrt{\y}})(\sqrt{\yji}-\overline{\sqrt{\y}}),\]
for
\begin{equation}
\overline{\sqrt{\y}}=\frac{1}{\abs{\dysY}}\sum_{\ipjpdysY}\sqrt{\yipjp}. \label{eq:mean-sqrt}
\end{equation}

\paragraph{Minimum}
An alternative generalization is to take the minimum of the two
values. For example,
if \[\genstat(\y)=\left(\sum_{\ijdysY}\yij,\sum_{\ijdysY,i<j}\min(\yij,\yji)\right)\t,\]
then
\begin{equation}\Ptheg(\Yyij|\Y\in\netsY\sij(\y))\propto\myexp{\curvpar_1\yij+\curvpar_2\min(\yij-\yji,0)}/\yij!.\label{eq:min-model}\end{equation}
Thus, a possible interpretation for this term is that the conditional
probability for a particular value of $\Yij$, $\yij$ is deflated by
$\myexp{\curvpar_2}$ for every unit by which $\yij$ is less than
$\yji$. In a sense, $\yji$ \enquote{pulls up} $\yij$ to its level and
vice versa.

\paragraph{Negative difference}
Generalizing the concept of similarity between $\yij$ and $\yji$ leads
to a statistic of difference between their values. We negate it so
that a positive coefficient value leads to greater mutuality. Then,
\begin{equation}\genstat(\y)=\left(\sum_{\ijdysY}\yij,\sum_{\ijdysY,i<j}-\abs{\yij-\yji}\right)\t,\label{eq:negdiff-model}\end{equation}
and 
\[\Ptheg(\Yyij|\Y\in\netsY\sij(\y))\propto\myexp{\curvpar_1\yij-\curvpar_2\abs{\yij-\yji}}/\yij!,\]
so the conditional probability of a particular $\yij$ is deflated
by $\myexp{\curvpar_2}$ for every unit difference from $\yji$, in
either direction. Thus, $\yji$ \enquote{pulls in} $\yij$ and vice
versa. Of course, other differences (e.g., squared difference) are also
possible.

We use the discrete change statistic to visualize the differences
between these variants in Figure~\ref{fig:mutuality-eff}, plotting the
$\curvpar_{\leftrightarrow}\changeij^{0\to\yij}\genstat_{\leftrightarrow}(\y)$
summand of
\[\log\frac{\Ptheg(\Yyij|\Y\in\netsY\sij(\y))}{\Ptheg(\Yij=0|\Y\in\netsY\sij(\y))}=\curvpar\cdot\changeij^{0\to \yij}\genstat(\y)\]
for each variant. Lastly, while the conditional distributions, and
hence the parameter interpretations for the minimum and the negative
difference statistic, are different, models induced by
\eqref{eq:min-model} and \eqref{eq:negdiff-model} are also
reparametrizations of each other: $\min(\yij,\yji) =
\half\left((\yij+\yji)-\abs{\yij-\yji}\right)$.
\begin{figure}
\begin{center}
  \includegraphics[scale=0.8]{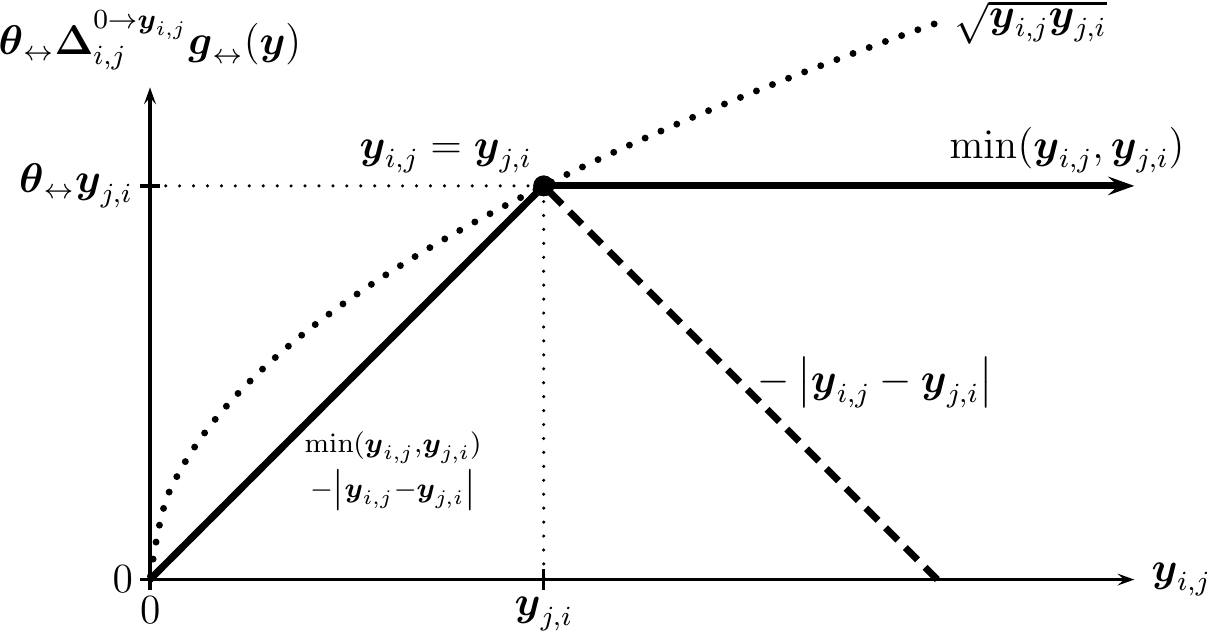}
\end{center} 
\caption{\label{fig:mutuality-eff}Effect of proposed mutuality
  statistics ($\genstat_{\leftrightarrow}$) with parameter
  $\curvpar_{\leftrightarrow}>0$ on the distribution of $\Yij$, given
  that $\Yji=\yji$. Whereas the $\min(\yij,\yji)$ statistic deflates
  the probabilities of those values of $\yij$ that are less than
  $\yji$, thus inflating all of those of $\yij$ above or equal to it,
  thus ``pulling $\Yij$ up'', the $-\abs{\yij-\yji}$ statistic
  deflates the probabilities in both directions away from $\yji$, thus
  inflating those that are the closest, ``pulling $\Yij$
  in''. $\sqrt{\yij\yji}$ inflates greater values of $\yij$ in
  general, inflating by more for greater $\sqrt{\yji}$.}
\end{figure}

\subsubsection{\label{sec:heterogeneity}Actor heterogeneity}
Another property found in social networks is that different
individuals have different overall propensities to have ties:
activity, popularity, and (undirected) sociality. Some of this
heterogeneity may be accounted for by exogenous covariates. For the
unaccounted-for heterogeneity, two major approaches have been used:
\emph{conditional}, in which actor-specific parameters are added to
the model to absorb it, with dyadic independence typically assumed
given these parameters, and \emph{marginal}, in which statistics are
added that represent the effects of heterogeneity on the overall
network features.  Examples of the conditional approach include the
very first exponential-family model for networks, the $p_1$, which
included a fixed effect for every actor \citep{holland1981efp}; and
the $p_2$ model and latent space models, which used using random
effects
\citep{duijn2004re,hoff2005bme,krivitsky2009rdd,mariadassou2010uls}. The
marginal approach includes the $k$-star statistics for $k\ge 2$
\citep{frank1986mg}, which, for a fixed network density, become more
prevalent as heterogeneity increases at the cost of often inducing
degeneracy; alternating $k$-stars and geometrically weighted degree
statistics \citep{snijders2006nse,hunter2006ice}, which attempt to
remedy the degeneracy of $k$-stars; and statistics such as the square
root degree activity/popularity, which sum of each actors' degree
taken to $3/2$ power, which also increases with greater heterogeneity,
but not as rapidly as $k$-stars do \citep{snijders2010isa}, avoiding
degeneracy. In the conditional approach, using fixed effects lacks
parsimony and using random effects creates a problem with a
doubly-intractable normalizing constant, beyond the scope of this
paper, so we develop a marginal approach here.

Actor heterogeneity may be viewed marginally as positive within-actor
correlation among the dyad values. Following the discussion in the
previous sections, we propose a form of pooled within-actor covariance
of variance-stabilized dyad values, scaled to the same magnitude as
the dyad sum:
\begin{equation}
\genstat(\y)=\sum_{i\in N}\frac{1}{n-2}\sum_{j,k\in \dysY_{i\to}\land j<k}(\sqrt{\yij}-\overline{\sqrt{\y}})(\sqrt{\y_{i,k}}-\overline{\sqrt{\y}}), \label{eq:within-actor-covariance}
\end{equation}
for $\dysY_{i\to}$ being the set of actors to who whom $i$ may have
ties ($\equiv \{j':(i,j')\in\netsY\}$) and $\overline{\sqrt{\y}}$
defined as in~\eqref{eq:mean-sqrt}. This statistic would increase with
greater out-tie heterogeneity, an analogous statistic could be
specified for in-tie heterogeneity, and dropping the directionality
would produce an undirected version of this statistic.

We have considered other variants, including the uncentered version,
in which each summand in \eqref{eq:within-actor-covariance} is simply
$\sqrt{\yij\y_{i,k}}$. We found that in undirected networks in
particular, such a model term can induce a degeneracy-like bimodal
distribution of networks. (This is likely because in undirected
networks, the positive dependence is not contained within each actor,
so subtracting $\overline{\sqrt{\y}}$ serves as a counterweight to
avert the \enquote{avalanche}.)

\subsubsection{\label{sec:triad-closure}Triad-closure bias}
We now turn to the question of how to represent triad-closure bias ---
friend-of-a-friend effects --- in count data. As with mutuality,
merely multiplying values of the dyads in a triad leads to a model
which cannot have positive triad closure bias. In addition, ERGM
sufficient statistics that take counts over triads often exhibit
degeneracy. \citep{schweinberger2011isd} For these reasons, we
describe a family of statistics that sum over dyads
instead. \citet{wyatt2010dlr} use a generalization of the curved
geometrically-weighted edgewise shared partners (GWESP) statistic
\citep{hunter2006ice}, though it is not clear whether it is suitable
for data with an infinite sample space. We thus describe a more
conservative family of statistics.

One term used to model triad closure in binary dynamic networks by
\citet{snijders2010isa} is the \emph{transitive ties} effect, the most
conservative special case of the GWESP \citep{hunter2006ice}
statistic. This statistic counts the number of ties $\pij$ such that
there exists at least one path of length 2 (\emph{two-path}) between
them --- a third actor $k$ such that $\y_{i,k}=\y_{k,j}=1$. (Unlike
the triangle count, each dyad may contribute at most $+1$ to the
statistic, no matter how many such $k$s exist.)

One generalization of this statistic to counts is
\begin{equation}\genstat(\y)=\sum_{\ijdysY}\min\left(\yij,\max_{k\in N}\left(\min(\y_{i,k},\y_{k,j})\right)\right).\label{eq:transitiveties}\end{equation}
Intuitively, define the strength of a two-path from $i$ to $j$ to be
the minimum of the values along the path. The statistic is then the
sum over the dyads $\pij$ of the minimum of the value of $\pij$ and
the value of the strongest two-path between. The interpretation is
thus somewhat analogous to that of the minimum mutuality statistic,
with $\yji$ replaced by $\max_{k\in N}(\min(\y_{i,k},\y_{k,j}))$. The
motivation for using minimum, as opposed to negative absolute
difference, to combine the two-path value with the focus dyad value is
that the intuitive notion of friend-of-a-friend effect that this
statistic embodies suggests that while the presence of a mutual friend
may increase the probability or expected value of a particular
friendship (i.e., \enquote{pull it up}), it should not limit it (i.e.,
\enquote{pull it in}) as an absolute difference would. These
interpretations are somewhat oversimplified: it is just as true that a
positive coefficient on this statistic results in $\yij$
\enquote{pulling up} the potential two-paths between $i$ and $j$.

In a directed network, \eqref{eq:transitiveties} would model transitive (hierarchical) triads, while
\[\genstat(\y)=\sum_{\ijdysY}\min\left(\yij,\max_{k\in N}\left(\min(\y_{j,k},\y_{k,i})\right)\right)\] 
would model cyclical (antihierarchical) triads.

The statistic \eqref{eq:transitiveties} is a fairly conservative one,
less likely to induce excessive dependence and bimodality, at the
cost of sensitivity. More generally, one may specify a triadic
statistic using three functions: first,
$v_{\text{2-path}}:\NN_0^2\to\RR$, how the \enquote{value} of a
two-path $i\to j \to k$ is computed from its constituent segments;
second, $v_{\text{combine}}:\RR^{n-2}\to\RR$, how the values of the
possible two-paths from $i$ to $j$ are combined with each other to
compute the strength of the pressure on $i$ and $j$ to close the triad
or increase their interaction; and third,
$v_{\text{affect}}:\NN_0\times\RR \to \RR$ how this pressure affects
$\Yij$. Given these,
\begin{equation}
\genstat(\y)=\sum_{\ijdysY}v_{\text{affect}}\left(\yij,v_{\text{combine}}\left(v_{\text{2-path}}(\y_{i,k},\y_{k,j})_{k\in N\setsub \{i,j\}}\right)\right).\label{eq:transitiveties-general}
\end{equation}

Thus, for example, one could set $v_{\text{combine}}$ to sum its
arguments rather than take their maximum, or one can replace taking
the minimum with taking a geometric mean.  We illustrate the
difference it makes in Section~\ref{sec:bkfrab}.

\section{\label{sec:example}Examples}

\subsection{\label{sec:zachary}Example 1: Social relations in a karate club}
In this application, we use a Poisson-reference ERGM to compare impacts
of social forces --- transitivity and homophily --- on the structure
of a valued network of interactions between members of a university
karate club. \citet{zachary1977ifm} reported observations of social
relations in a university karate club with membership that varied
between 50 and 100. The actors --- 32 ordinary club members and
officers, the club president (\enquote{John A.}), and the part-time
instructor (\enquote{Mr. Hi}) --- were the ones who consistently
interacted outside of the club. Over the course of the study, the club
divided into two factions, and, ultimately, split into two clubs, one
led by Hi and the other by John and the original club's officers. The
split was driven by a disagreement over whether Hi could unilaterally
change the level of compensation for his services. We pose a similar
question to \citet{goodreau2008bff}: is the structure at the time of
observation driven by faction allegiance or by transitivity
(\enquote{friend-of-a-friend} effects)?

\citeauthor{zachary1977ifm} identifies the faction with which each of
the 34 actors was aligned and how strongly and reports, for each pair
of actors, the count of social contexts in which they interacted. The
8 contexts considered were academic classes at the university; Hi's
private karate studio in his night classes; Hi's private karate studio
where he taught on weekends; student-teaching at Hi's studio; the
university rathskeller (bar) located near the karate club; a bar
located near the university campus; open karate tournaments in the
area; and intercollegiate karate tournaments. The highest number of
contexts of interaction for a pair of individuals that was observed
was 7. The network is visualized in Figure~\ref{fig:zach-plot}.
\begin{figure}
  \begin{center}
  \includegraphics[trim=0 108 0 108,width=.8\textwidth,clip=true]{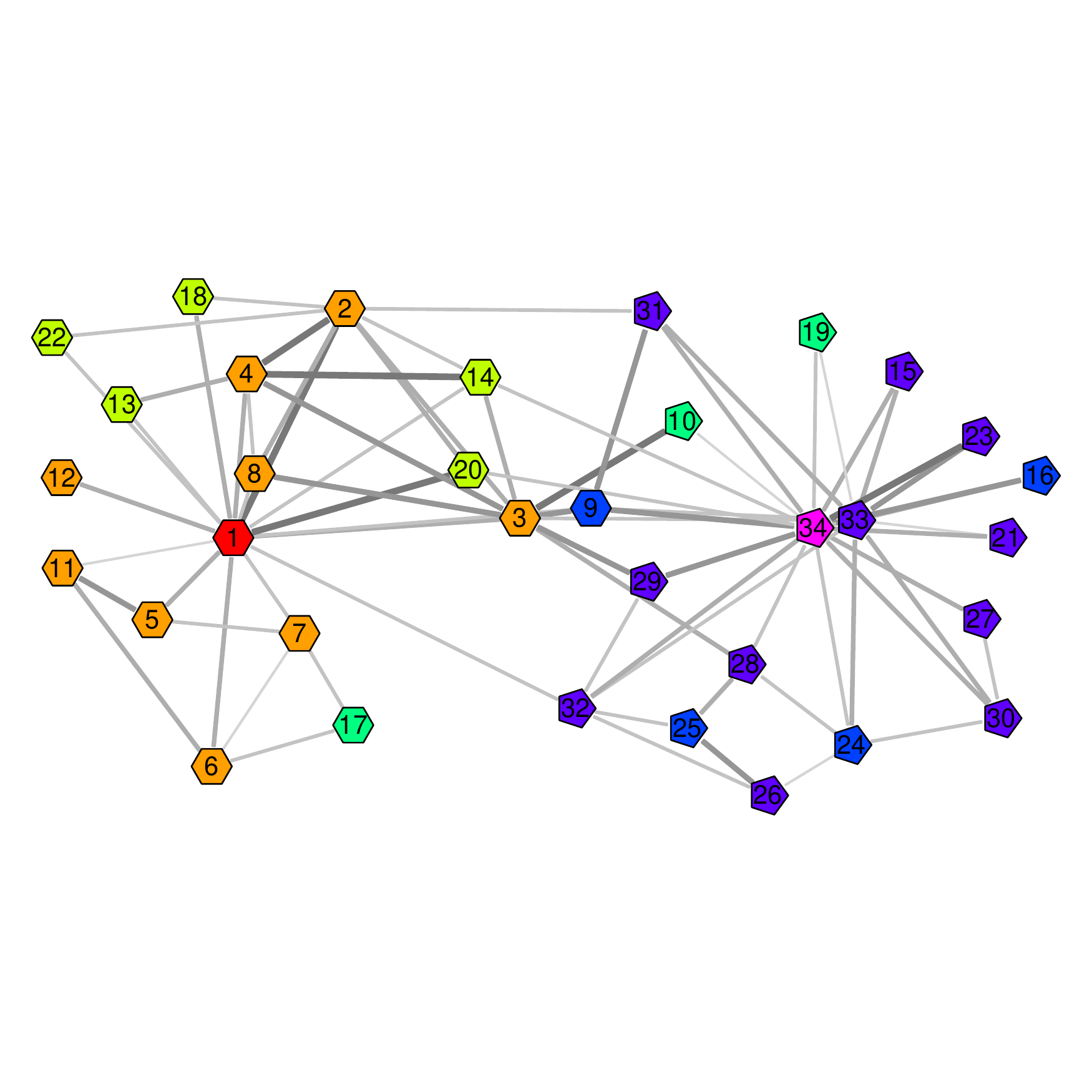}
    \begin{tabular}{lccc}
    Leader & \textcolor[HTML]{FF0000}{Mr. Hi} & &  \textcolor[HTML]{FF00FF}{John A.}\\
    Alignment & \textcolor[HTML]{FF9F00}{Strong} \textcolor[HTML]{BFFF00}{Weak} & \textcolor[HTML]{00FF80}{Neutral}  & \textcolor[HTML]{0040FF}{Weak} \textcolor[HTML]{6000FF}{Strong}\\
    Final & \includegraphics[height=1em]{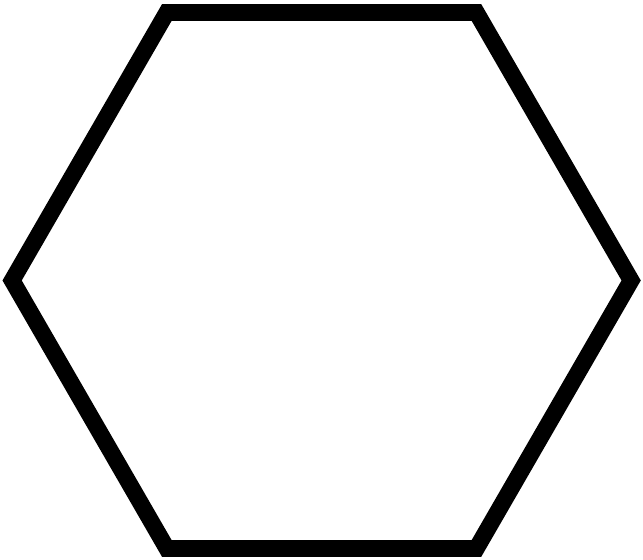} & & \includegraphics[height=1em]{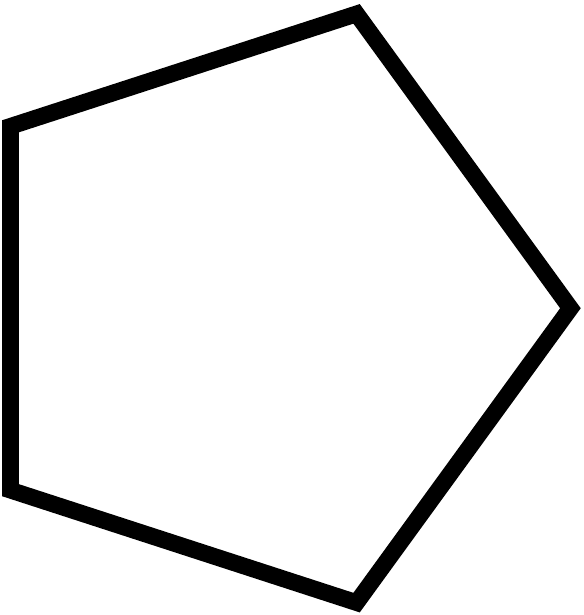}
  \end{tabular}
  \end{center}
\caption{\label{fig:zach-plot} The sociogram of the \citet{zachary1977ifm} karate club network. The color of each plotting symbol shows each actor's faction alignment and its shape shows which of the two clubs the actor ultimately joined. Darker and thicker lines correspond to more social context of interaction. }
\end{figure}

We begin with a Poisson-reference ERGM. Empirically, this network is
more sparse than a Poisson density for dyad values would suggest: the
mean number of dyadwise contexts of interaction
($\sum_{\ijdysY}\yij/\abs{\dysY}$) is $0.41$, for which a Poisson
distribution predicts an expected density
($\E(\sum_{\ijdysY}\I{\Yij>0}/\abs{\dysY})$) of $0.34$, whereas the
observed density is $0.14$. Given that two individuals interact, the
interaction for a given pair of individuals is likely to be dependent
across the social contexts counted so the counts are likely to be
over- or under-dispersed. Thus, as a baseline, we model the values as
a zero-modified Conway--Maxwell--Poisson \citep{shmueli2005udf}
distribution using the following sufficient statistics:
\begin{description}[partopsep=0pt,topsep=0pt,parsep=0pt,itemsep=0pt]
  \item[baseline propensity to have ties:] number of dyads with nonzero value;
  \item[baseline intensity of interactions:] sum of dyad values; and
  \item[CMP dispersion:] a statistic of the form \eqref{eq:CMP-term}.
\end{description}

In modeling the structure of the interactions, we represent
differential propensity of the faction leaders to interact, the effect
of differences in faction membership, and triad-closure bias using the
following sufficient statistics:
\begin{description}[partopsep=0pt,topsep=0pt,parsep=0pt,itemsep=0pt]
  \item[intensity of Hi's interaction:] sum of dyad values incident on Hi;
  \item[intensity of John's interaction:] sum of dyad values incident on John;
  \item[similarity (negative difference) in faction membership:] a
    statistic of the form \eqref{eq:poisreg-stat} with
    $\xijk=-\abs{m_i-m_j}$, where $m_i$ is the faction membership
    code of actor $i$; and
  \item[transitivity of intensities:] the statistic
    \eqref{eq:transitiveties}.
\end{description}
Faction memberships $m_i$ are coded as follows: strongly Hi's as $-2$,
weakly Hi's as $-1$, neutral as $0$, weakly John's as $+1$, and
strongly John's as $+2$. We fit three models: the full model, with all
of the above-described terms, the model excluding transitivity
(\enquote{Faction}), and the model excluding faction membership
(\enquote{Transitivity}).

\begin{table}
\caption{\label{tab:zach-fit}Results from fitting the models to Karate Club network}
\begin{center}
\begin{tabular}{lrrr}
  \hline
  & \multicolumn{3}{c}{Estimates (Std. Errors)} \\
  Term & \centercell{Faction} & \centercell{Transitivity} & \centercell{Full} \\ 
  \hline
  Dispersion & \coefsig{-2.55}{0.57} & \coefsig{-1.87}{0.61} & \coefsig{-2.33}{0.60} \\
  Ties & \coefsig{-7.76}{0.99} & \coefsig{-7.29}{1.04} & \coefsig{-7.54}{1.01} \\ 
  Baseline intensities & \coefsig{3.97}{0.68} & \coefsig{2.88}{0.75} & \coefsig{3.64}{0.74} \\
  Hi's intensities & \coefsig{0.80}{0.15} & \coefsig{0.50}{0.12} & \coefsig{0.71}{0.15} \\ 
  John's intensities & \coefsig{0.80}{0.14} & \coefsig{0.54}{0.12} & \coefsig{0.72}{0.16} \\
  Faction similarity & \coefsig{0.27}{0.04} &  & \coefsig{0.25}{0.04} \\
  Transitivity &  & \coefsig{0.21}{0.09} & \coefnsig{0.11}{0.09} \\ 
   \hline
\end{tabular}
\end{center}
{
\parbox{1\textwidth}{
Coefficients statistically significant at $\alpha=0.05$ are \textbf{bolded}.\\
Standard errors incorporate the uncertainty introduced by approximating of the likelihood using MCMC \citep{hunter2006ice}.}
}
\end{table}
Table~\ref{tab:zach-fit} gives the results for the three fits. MCMC
diagnostics, described by \citet{goodreau2008st}, show adequate mixing
and networks simulated from these fits have, on average, statistics
equal to the observed sufficient statistics. The CMP dispersion
estimates for all three models are negative and highly significant,
very far from the non-open boundary of the parameter space at
$\curvpar_k\le 1$, so the lack of steepness is unlikely to be
problematic in this case. The estimated value of the dispersion
parameter for the full model ($-2.33$) suggests strong underdispersion
relative to zero-modified Poisson and the rest of the model: it
implies that the estimated \enquote{denominator} is
$(\yij!)^{1-(-2.33)}=(\yij!)^{3.33}$, rather than Poisson's
$(\yij!)^{1}$. Highly negative CMP coefficients may also be
interpreted as the model being an overfit.

There is a highly significant negative coefficient on the baseline
propensity for ties. An interpretation for this is that, from the
point of view of a single dyad, the probability of a given pair of
actors having a tie is deflated, but if they do have a tie, it is
likely to be across multiple social contexts. Both faction leaders
appear to have greater overall propensities to interact than the other
club members, and, interestingly, they appear to have similar effect
sizes to each other.

Taken separately, the faction similarity effect is highly
statistically significant and positive, indicating a positive faction
cohesion. The transitivity effect is significant by the Wald test, but
a Monte Carlo test gives its one-sided (since only positive
transitivity is of interest) $P$-value as $0.11$. Put together, the
transitivity loses any potential significance. (Notably, the estimated
correlation between their parameter estimates is $-0.34$.) This
suggests that they are explaining the same aspect of the network
structure, but that faction allegiance is the much stronger
explanation. Though factions may themselves be endogenous due to
influence through social relations or, as \citeauthor{zachary1977ifm}
concludes, the two processes reinforced each other over time, at
observation time, faction allegiance explains network structure better
than transitivity.

\subsection{\label{sec:bkfrab}Example 2: Interactions in a fraternity}
In a series of studies in the 1970s, \citet{bernard19791980ias}
assessed accuracy of retrospective sociometric surveys in a number of
settings, including a college fraternity whose 58 occupants had all
lived there for at least three months. To record the true amounts of
interaction, for several days, unobtrusive observers were sent to
periodically walk through the fraternity to note students engaged in
conversation. Obtaining these network data from
\citet{batagelj2006pd}, we model these observed pairwise interaction
counts.

The raw distribution of counts, given in
Figure~\ref{fig:bkfrab-summ-hist}, appears to be strongly
overdispersed relative to Poisson, and, indeed, relative to the
geometric distribution: the mean of counts is $2.0$, while their
standard deviation (not variance) is $3.4$. At least some of this is
due to actor heterogeneity: the square root of the within-actor
variance of the counts is $3.1$. Excluding extreme observations (all
values over 30) does not make a qualitative difference. (The
statistics are $1.9$, $3.0$, and $2.8$, respectively.) Nor does there
appear a natural place to threshold the counts to produce a binary
network. (See Figure~\ref{fig:bkfrab-summ-threshold}.)
\begin{figure}
  \centering
  \subfloat[Count distribution]{\label{fig:bkfrab-summ-hist}\includegraphics[width=0.45\textwidth]{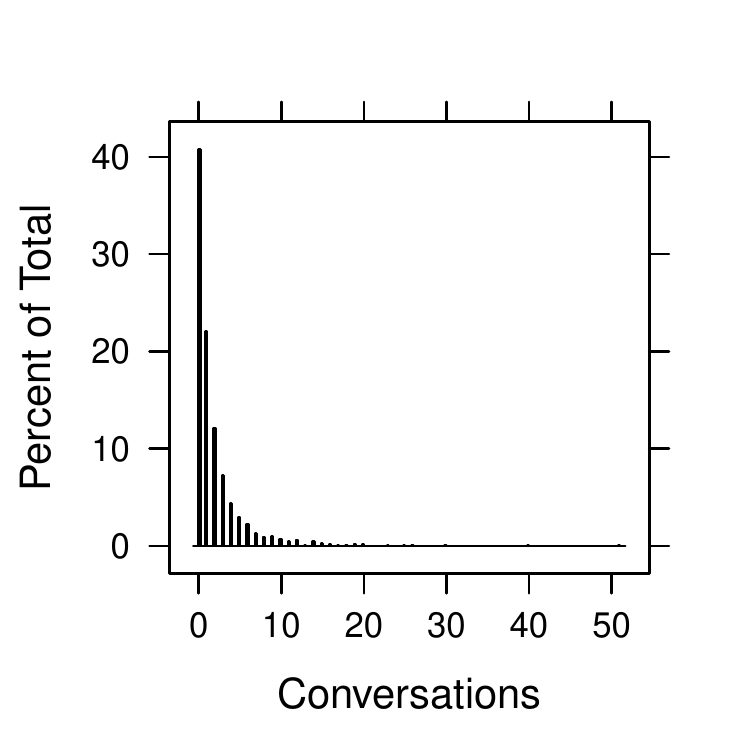}}\quad\subfloat[Effect of thresholding]{\label{fig:bkfrab-summ-threshold}\includegraphics[width=0.45\textwidth]{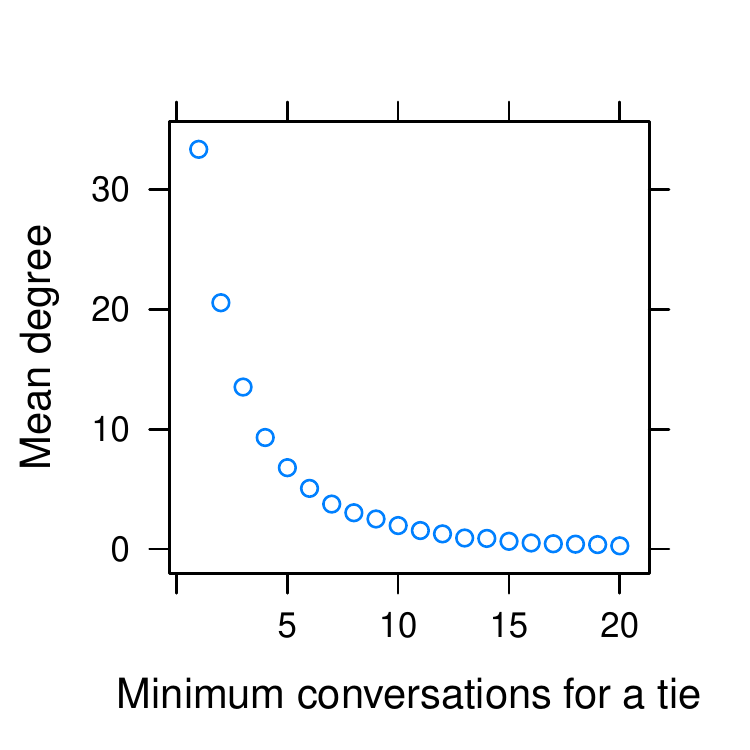}}
  \caption{\label{fig:bkfrab-summ} Conversation count summaries for
    Bernard and Killworth's fraternity network }
\end{figure}
We thus model the baseline shape of the distribution of counts using the following terms:
\begin{description}[partopsep=0pt,topsep=0pt,parsep=0pt,itemsep=0pt]
\item[baseline propensity to have ties:] number of dyads with nonzero value;
  \item[baseline intensity of interactions:] sum of dyad values; and
  \item[underdispersion:] the statistic \eqref{eq:sqrtsum-model}.
\end{description}
(We have also attempted to use CMP (via~\eqref{eq:CMP-term}) but found the process
to be unstable due to the greater-than-geometric level of dispersion.)

Little was recorded about the social roles of the fraternity members,
so we consider the effects of endogenous social forces:
\begin{description}[partopsep=0pt,topsep=0pt,parsep=0pt,itemsep=0pt]
\item[actor heterogeneity:] the undirected version of
  \eqref{eq:within-actor-covariance};
  \item[transitivity of intensities:] the statistic
    \eqref{eq:transitiveties}.
\end{description}
\cite{faust2007vls}, in particular, found that in many empirical
networks, much of the apparent triadic effects are accounted for by
variations in degree distribution and other lower-order effects. Thus,
we consider four models: baseline shape only (B), baseline with
heterogeneity (BH), baseline with transitivity (BT), and all terms
(BHT), to explore this concept in a valued setting.

\begin{table}
  \caption{\label{tab:bkfrab-fit}Results from fitting the models to Bernard and Killsworth's fraternity network}
\begin{center}
{\small
\begin{tabular}{lrrrr}
  \hline
  & \multicolumn{4}{c}{Estimates (Std. Errors)} \\
  Term & \centercell{B} & \centercell{BH} & \centercell{BT} & \centercell{BHT} \\ 
  \hline
  Ties & \coefsig{5.60}{0.21} & \coefsig{4.96}{0.17} & \coefsig{6.24}{0.21} & \coefsig{4.98}{0.17} \\ 
  Intensity & \coefsig{3.65}{0.05} & \coefsig{3.13}{0.06} & \coefsig{3.40}{0.07} & \coefsig{3.12}{0.06} \\
  Underdispersion & \coefsig{-9.71}{0.22} &\coefsig{-8.23}{0.20} & \coefsig{-10.52}{0.22} & \coefsig{-8.26}{0.19} \\
  Heterogeneity & & \coefsig{1.48}{0.06} &  & \coefsig{1.46}{0.07} \\
  Transitivity & & & \coefsig{0.46}{0.05} & \coefnsig{0.03}{0.04} \\ 
  \hline
\end{tabular}
}
\end{center}
{\footnotesize
\parbox{1\textwidth}{
Coefficients statistically significant at $\alpha=0.05$ are \textbf{bolded}.\\
Standard errors incorporate the uncertainty introduced by approximating of the likelihood using MCMC \citep{hunter2006ice}.}
}
\end{table}
We report the model fits in Table~\ref{tab:bkfrab-fit}. MCMC
diagnostics, described by \citet{goodreau2008st}, show adequate mixing
and unimodal distributions of sufficient statistics, and networks
simulated from these fits have, on average, statistics equal to the
observed sufficient statistics. The baseline dyadwise distribution
terms are difficult to interpret, but the highly negative coefficient
on underdispersion suggests a a strong degree of overdispersion, as
expected. Some of this overdispersion appears to be absorbed by
modeling actor heterogeneity, however. There are indications a high
degree of heterogeneity in individuals' interactions, over and above
that expected for even the overdispersed baseline
distribution. (Monte Carlo $P\text{-val.}<0.001$ based on
10,000 draws.)

Without accounting for actor heterogeneity (i.e., Model BT), there
appears to be a strong transitivity effect --- a friend of a friend is
a friend --- and the Monte Carlo test confirms this with a similar
$P$-value. However, if actor heterogeneity is accounted for, the
transitivity effects vanish (simulated one-sided
$P\text{-val.}=0.43$), suggesting that the underlying social process
is better explained by a relatively small number of highly social
individuals whose ties to each other and to (less social) third
parties create excess transitive ties for the overall amount of
interaction observed. At the same time, if, instead of using
\eqref{eq:transitiveties} as the test statistic, we use a less
conservative statistic of the form \eqref{eq:transitiveties-general}
with $v_{\text{2-path}}(x_1,x_2)=\sqrt{x_1x_2}$ (geometric mean),
$v_{\text{combine}}(x_1,\dotsc,x_{n-2})=\sum_{k=1}^{n-2}x_k$, and
$v_{\text{affect}}(x_1,x_2)=\sqrt{x_1x_2}$, the effect's significance
seems to increase (one-sided $P\text{-val.}=0.07$). However, when we
attempted to fit the model with this effect, the process exhibited the
degeneracy-like bimodality. This suggests that there is a trade-off
between stability and power to detect subtle effects.

\section{\label{sec:discussion}Discussion}

We have generalized the exponential-family random graph models to
networks whose relationships are unbounded counts, explored the issues
that arise when generalizing, and proposed ways to model several
common network features for count data. We demonstrated our
development by analyzing two very different networks to examine the
interaction of friend-of-a-friend effects with homophily and
individual heterogeneity.

This paper focused on modeling counts. More generally, one can express
a \wergm{} by replacing the set of possible dyad values $\NN_0$ by a
more general set $\dyadvals$ and replacing $\h(\y)$ with a more
general $\sigma$-finite measure space $(\netsY,\sigY,\Mref)$ with
reference measure $\Mref$, then postulating a probability measure
$\Mtheg$ with Radon-Nikodym derivative of $\Mtheg$ with respect to
$\Mref$,
\begin{equation*}
  \dthegdref(\y)=\frac{\myexp{\innerprod{\natpar}{\genstat(\y)}}}{\cHeg(\curvpar)},\label{eq:nbergm}
\end{equation*}
\citetext{\citealp[pp.~115--116]{barndorffnielsen1978ief};
  \citealp[pp.~1--2]{brown1986fse}} with the normalizing constant
\[\cHeg(\curvpar)=\int_\netsY\myexp{\innerprod{\natpar}{\genstat(\y)}}d\Mref(\y).\]
For binary and count data, and discrete data in general, $\Mref$ could
be specified as a function relative to the counting measure, while for
continuous data, it could be defined with respect to the Lebesgue
measure. Still, as with count data, the shape of this function would
need to be specified.

Other scenarios might call for more complex specifications of the
reference measure. Some network data, such as measurements of duration
of conversation \citep{wyatt2010dlr} and international trade volumes
\citep{westveld2011mem} are continuous measurements except for having
a positive probability of two actors not conversing at all or two
countries having no measured trade. \citeauthor{westveld2011mem} use a
normal distribution to model the log-transformed trade volume,
imputing $0=\log(1)$ for $0$ observed trade volumes (all nonzero trade
volumes being greater than $1$ unit), and they note this issue and address it
by pointing out that in their (latent-variable) model, an impact of
such an outlier would be contained. \Wergms{} may provide a more
principled approach by specifying a semicontinuous $\Mref$, such as
one that puts a mass of $1/2$ on $0$ and $1/2$ on Lebesgue measure
on $(0,\infty)$.

We have also focused on data that do not impose any constraints on the
sample space: $\netsY\equiv\dyadvals^\dysY$. But, some types of
network data, such as those where each actor ranks the others
\citep{newcomb1961ap} may be viewed as imposing a more complex
constraint on sample spaces: setting $\dyadvals=\{1..n-1\}$ and
constraining $\netsY$ to ensure that each ego assigns a unique rank to
each alter gives a sample space of permutations that
could, with a counting measure, serve as the reference measure for an
ERGM on rank data. These, and other applications are a subject for
ongoing and future work.

This paper focuses on models for cross-sectional networks, where a
single snapshot of relationship states or relationships aggregated
over a time period are observed. For longitudinal data, comprising
multiple snapshots of networks over the same actors over time, binary
ERGMs have been used as a basis for discrete-time models for network
tie evolution by \citet{robins2001rgm}, \citeauthor{wyatt2009dmn}
\citeyearpar{wyatt2009dmn, wyatt2010dlr}, \citet{hanneke2010dtm}, and
\citet{krivitsky2010smd}.
\Wergms{} can be directly applied to the discrete temporal
ERGMs of \citet{hanneke2010dtm} although their adaptation to the work
of \citet{krivitsky2010smd} may be less straightforward, especially if
the benefits to interpretability of the separable models are to be
retained.

In practice, networks are not always observed
completely. \citet{handcock2010msn} develop an approach to ERGM
inference for partially observed or sampled binary networks. It would
be natural to extend this approach to valued networks and \wergms{}.

Some methods for assessing a network model's fit, particularly MCMC
diagnostics \citep{goodreau2008st} can be used with little or no
modification. Others, like the goodness-of-fit methods of
\citet{hunter2008gfs} may require development of characteristics
meaningful for valued networks. It may also be possible to extend the
stability criteria of \citet{schweinberger2011isd} to models with
infinite sample spaces.

\section*{Acknowledgments}
The author thanks Mark S. Handcock for helpful discussions and
comments on early drafts; Stephen E. Fienberg for his feedback on this
manuscript; and Michael Schweinberger, David R. Hunter, Tom
A. B. Snijders, and Xiaoyue Niu for their comments and advice. This
research was supported by Portuguese Foundation for Science and
Technology Ci\^{e}ncia 2009 Program, ONR award N000140811015, and NIH
award 1R01HD068395-01.

\bibliographystyle{plainnat}
\addcontentsline{toc}{section}{References}
\bibliography{ERGMs_for_Valued_Networks}

\newpage

\appendix
\section{\label{app:pois-MH} A sampling algorithm for a Poisson-reference ERGM}
We use a Metropolis-Hastings sampling algorithm
(Algorithm~\ref{alg:pois-MH}) to sample from a Poisson-reference ERGM,
using a Poisson kernel with its mode at the present value of a dyad
and, occasionally (with a specified probability $\pi_0$), proposing a
jump directly to $0$. Because, as we discuss in
Section~\ref{sec:count-thresholded}, counts of interactions are often
zero-inflated relative to Poisson, setting $\pi_0>0$ can be used to
speed-up mixing. For highly overdispersed distributions, a Poisson
kernel may be trivially replaced by a geometric or even
negative-binomial kernel.

This algorithm selects the dyad on which the jump is to be proposed
at random. A possible improvement to this algorithm would be to adapt
to it the tie-no-tie (TNT) proposal \citep{morris2008ser}, which
optimizes sampling in sparse (zero-inflated) networks by focusing on
dyads which have a nonzero values.
\begin{algorithm}
  \caption{\label{alg:pois-MH}Sampling from a Poisson-reference ERGM
    with no constraints, optimized for zero-inflated distributions}
\begin{flushleft}
  \begin{description}[partopsep=0pt,topsep=0pt,parsep=0pt,itemsep=0pt]
  \item[Let:]\
    \begin{description}[partopsep=0pt,topsep=0pt,parsep=0pt,itemsep=0pt]
    \item[$\RandomChoose(A)$] return a random element of a set $A$
    \item[$\Uniform(a,b)$] return a random draw from the
      $\Uniform(a,b)$ distribution
    \item[$\Poisson_{\ne y}(\lambda)$] return a random draw from the
      $\Poisson(\lambda)$ distribution, conditional on not drawing $y$
    \item[$p(y^*;y)$] $=\frac{\myexp{-\left(y+\half\right)}\left(y+\half\right)^{y^*}/y^*!}{1-\myexp{-\left(y+\half\right)}\left(y+\half\right)^y/y!}$, the
      pmf of a $\Poisson_{\ne y}(y+\half)$ draw
    \end{description}
  \end{description}
\end{flushleft}
\begin{algorithmic}[1]
  \REQUIRE $\y^{(0)}\in\netsY$, $T$ sufficiently large, $\dysY$, $\genstat$, $\cnmap$, $\pi_0\in [0,1)$
  \ENSURE a draw from the specified Poisson-reference ERGM
  \FOR{$t \gets 1..T$}
  \STATE $\pij \gets \RandomChoose(\dysY)$ \COMMENT{Select a dyad at random.}
  \IF{$\yij\ne 0 \land \Uniform(0,1)<\pi_0$}
  \STATE $y^*\gets 0$ \COMMENT{Propose a jump to $0$ with probability $\pi_0$.}
  \ELSE
  \STATE $y^* \gets \Poisson_{\ne \yij^{(t-1)}}\left(\yij^{(t-1)}\right)$ \COMMENT{Propose a jump to a new value.}
  \ENDIF
  \STATE $q \gets 
  \begin{cases}
    \frac{\pi_0+(1-\pi_0)p(0;y^*)}{p(y^*;0)} & \yij^{(t-1)}=0 \\
    \frac{p(\yij^{(t-1)};0)}{\pi_0+(1-\pi_0)p(0;\yij^{(t-1)})}& \yij^{(t-1)}\ne 0 \land y^*=0 \\
    \frac{\cancel{(1-\pi_0)}p(\yij^{(t-1)};y^*)}{\cancel{(1-\pi_0)}p(y^*;\yij^{(t-1)})} & \text{otherwise}
  \end{cases}$
  \STATE $r\gets q \times \frac{\yij^{(t-1)}!}{y^*!} \times \myexp{\innerprod{\natpar}{\changeij^{\yij^{(t-1)} \to y^*}\left(\y^{(t-1)}\right)}}$
  \IF{$\Uniform(0,1)<r$}
  \STATE $\y^{(t)}\gets\y^{(t-1)}_{\pij=y^*}$ \COMMENT{Accept the proposal.}
  \ELSE
  \STATE $\y^{(t)}\gets\y^{(t-1)}$ \COMMENT{Reject the proposal.}
  \ENDIF
  \ENDFOR
  \RETURN $\y^{(T)}$
\end{algorithmic}
\end{algorithm}

\section{\label{app:CMP-steep} Non-steepness of the Conway--Maxwell--Poisson family}
Expressed in its exponential-family canonical form, a random variable
$X$ with the Conway--Maxwell--Poisson distribution has the pmf
\[\Pteg(X=x)=\frac{\myexp{\curvpar_1 x+\curvpar_2 \log(x!)}}{\ceg(\curvpar)},\ x\in\NN_0\]
with the normalizing constant
\[\ceg(\curvpar)=\sum_{x'=0}^\infty\frac{\myexp{\curvpar_1 x'+\curvpar_2 \log(x'!)}}{\ceg(\curvpar)}.\]
\begin{thm}
The Conway--Maxwell--Poisson family is not regular.
\end{thm}
\begin{proof}
The natural parameter space of CMP is
\[\natcurvpars=\{\curvpar'\in\RR^2:\curvpar_2<0 \lor (\curvpar_2=0 \land \curvpar_1<0)\}\]
\citep{shmueli2005udf}. Due to the boundary at $\curvpar_2=0$,
$\natcurvpars$ is not an open set, and hence the family is not regular
\citep[p. 2]{brown1986fse}.
\end{proof}

\begin{thm}
The Conway--Maxwell--Poisson family is not steep.
\end{thm}
\begin{proof}
  A necessary and sufficient condition for a non-regular exponential
  family to be steep is that
  \[\forall_{\curvpar\in \natcurvpars\setsub\natcurvpars\interior} \Eteg(\lVert \genstat(X)\rVert)=\infty,\]
  where $\natcurvpars\interior$ is the open interior of
  $\natcurvpars$, and their set difference is thus the non-open
  boundary of the natural parameter space that is contained within it.
  \citep[Proposition 3.3, p. 72]{brown1986fse} For CMP, this boundary
  \[\natcurvpars\setsub\natcurvpars\interior=\{\curvpar'\in\RR^2: \curvpar_2=0 \land \curvpar_1<0\}.\]

  There, $X\sim \Geometric(p=1-\myexp{\curvpar_1})$. Noting that $X\ge
  0$ a.s., $\log(X!)\ge 0$ a.s., and $\log(x!)\le
  (x+1)\log\left(\frac{x+1}{e}\right)+1$,
\begin{align*}
  \Eteg(\lVert \genstat(X)\rVert)&=\E_{\Geometric(p=1-\myexp{\curvpar_1})}(\rVert [X, \log(X!)]\t\rVert)\\
  & \le \E_{\Geometric(p=1-\myexp{\curvpar_1})}\left(X+\log(X!)\right)\\
  & \le \E_{\Geometric(p=1-\myexp{\curvpar_1})}\left(X+(X+1)\log\left(\frac{X+1}{e}\right)+1\right)\\
  & \le \E_{\Geometric(p=1-\myexp{\curvpar_1})}\left(X+ (X+1)^2 + 1\right)\\
  & < \infty,
\end{align*}
since the first and second moments of the geometric distribution are
finite. Therefore, CMP is not steep.
\end{proof}
Because the non-steep boundary corresponds to the most dispersed
distribution that CMP can represent, maximum likelihood estimator
properties for data which are highly overdispersed are not guaranteed.
\end{document}